\documentclass{IEEEtran4PSCC}

\ifCLASSINFOpdf
\else
\fi

\usepackage[cmex10]{amsmath}
\usepackage{color}
\usepackage{epsfig}
\usepackage{subfig}
\usepackage{mathrsfs}
\usepackage{booktabs}
\usepackage{graphicx}
\usepackage{epstopdf}
\usepackage{multirow}
\usepackage{amssymb}
\usepackage{amsmath}
\usepackage{amsthm}
\usepackage{fancyhdr}
\usepackage{tabularx}
\usepackage{xfrac}

\newtheorem{theorem}{Theorem}[section]

\newtheorem{prop}{Proposition}[section]

\DeclareMathOperator*{\argmin}{argmin}

\begin{document}

\title{A Robust Optimization Approach for Terminating the Cascading Failure of Power Systems}

\author{
\IEEEauthorblockN{Chao Zhai} 
\IEEEauthorblockA{School of Electrical and Electronic Engineering \\
Nanyang Technological University \\
50 Nanyang Avenue, Singapore 639798 \\
Email: zhaichao@amss.ac.cn}
}

\maketitle

\begin{abstract}
Due to uncertainties and the complicated intrinsic dynamics of power systems, it is difficult to predict the cascading failure
paths once the cascades occur. This makes it challenging to achieve the effective power system protection against cascading blackouts.
By incorporating uncertainties and stochastic factors of the cascades, a Markov chain model is developed in this paper to predict the cascading failure paths of power systems. The transition matrix of Markov chain is dependent on the probability of branch outage
caused by overloads or stochastic factors. Moreover, a robust optimization formulation is proposed to deal with the cascading blackouts by shedding load optimally for multiple cascading failure paths with relatively high probabilities. Essentially, it can be converted to the best approximation problem. Thus, an efficient numerical solver based on Dykstra's algorithm is employed to deal with the robust optimization problem. In theory, we provide a lower bound for the probability of preventing the cascading blackouts of power systems. Finally, the proposed approach for power system protection is verified by a case study of IEEE 118 bus system.
\end{abstract}



\section*{Nomenclature}
\begin{align*}
&\text{$P$}       &&\text{Transition matrix of Markov chain} \\
&\text{$p_{ij}$}  &&\text{Transition probability from State $i$ to State $j$} \\
&\text{$x^k$}     &&\text{Probability distribution vector over states } \\
&\text{~~}        &&\text{at the $k$-th cascading step} \\
&\text{$s^k$}     &&\text{Random vector describing the power system} \\
&\text{~~}        &&\text{state at the $k$-th cascading step} \\
&\text{$s_l^k$}   &&\text{Random variable describing the connection state}   \\
&\text{~~}        &&\text{of the $l$-th branch at the $k$-th cascading step} \\
&\text{$P_b^k$}   &&\text{Net active power vector injected to buses} \\
&\text{~~}        &&\text{at the $k$-th cascading step} \\
&\text{$P_e^k$}   &&\text{Power flow vector at the $k$-th cascading step} \\
&\text{$\lambda_{l}$} &&\text{Outage probability of the $l$-th branch} \\
&\text{$P_{X^{(i)}}(d)$} &&\text{Projection of the point $d$ onto the set $X^{(i)}$} \\
&\text{$m$}          &&\text{Number of buses in power systems} \\
&\text{$n$}          &&\text{Number of branches in power systems} 
\end{align*}
\begin{align*}
&\text{$h$}          &&\text{Number of cascading steps during the cascades} \\
&\text{$I_n$}        &&\text{Set of positive integers, i.e. $\{1, 2, ... n\}$} \\
&\text{$p_{over}$}   &&\text{Probability vector of branch outage due to} \\
&\text{~~}           &&\text{overloads} \\
&\text{$p_{hidden}$} &&\text{Probability vector of branch outage due to} \\
&\text{~~}           &&\text{hidden failures} \\
&\text{$p_{cont}$}   &&\text{Probability vector of branch outage due to} \\
&\text{~~}           &&\text{contingencies} \\
\end{align*}

\section{Introduction}
The past several decades have witnessed major blackouts of power systems in the world, such as the 1999 Southern Brazil Blackout, the 2003 US-Canada Blackout and the 2012 India Blackout, which have affected millions of people and caused huge economic losses \cite{mcl09}. According to technical reports on major power outages, power system blackouts normally go through five stages: precondition, initiating event, cascading events, final state and restoration \cite{lu06}. The successful prevention of cascading blackouts relies on the effective identification of disruptive initiating events and the accurate prediction of cascading failure paths. Here, a cascading failure path refers to the sequence of branch outages during the cascades, and it describes how the failure propagates throughout power systems. However, the uncertainties and interdependencies of diverse components make it difficult to predict the cascading failure paths once the cascade is triggered. As a result, it is a great challenge for protecting power systems against cascading blackouts.

The insights into the cascading sequences of power grids play an important role in designing an effective protection scheme to prevent cascading blackouts. The conventional protection of power systems largely relies on the protective relays for directly severing overloaded branches or the predetermined protection schemes to isolate the faulted components in power systems during emergencies,
such as under-voltage load shedding \cite{arn97}, under-frequency load shedding \cite{sur67,sta74} and the special protection scheme (SPS) \cite{hew04,and96}.  According to the definition from the North American Electric Reliability Council (NERC), an SPS is designed to detect abnormal system conditions and take preplanned, corrective action to provide acceptable system performance. It normally results in the decomposition of the whole system into several islands in order to isolate the faulted components \cite{sps98}.
Actually, each power system has its own emergency control practices and operating procedures, which are dependent on the different operating conditions, characteristics of the system \cite{beg05}. This implies that the operating procedure for each power system is unique. Although the effectiveness of conventional protection methods has been demonstrated in practice, they are actually subject to multiple limitations in terms of compatibility and universality for various contingencies that could trigger the cascades. For instance, the SPS is designed based on a number of specific scenarios, characterized by certain abnormal stresses, following which the system will collapse. A fault which is not contained in such a list of scenarios, however, may not be covered by the SPS. For this reason, it would be desirable to develop a disturbance-related real-time protection scheme that is able to prevent the degradation of power systems during cascades. The introduction of phasor measurement units (PMUs) in power systems allows to develop the real-time protection scheme by estimating the system states for emergency control \cite{nuq05}.
For example, \cite{zhai19} proposes a model predictive approach to prevent the cascading blackout of power systems by predicting the cascading failure path and taking the corresponding remedial actions in time. Nevertheless, the remedial action is computed based on a deterministic cascading failure path without allowing for the uncertainties and stochastic factors such as hidden failure, contingencies and malfunction of protective relays.

To address this problem, we propose a Markov chain model to describe the effect of uncertainties and stochastic factors on the cascading failure paths in this paper. The transition matrix of Markov chain is dependent on the probability of branch outage that is caused by overloads or other relevant stochastic factors. This Markov chain model allows to predict multiple cascading failure paths of relatively high probability. Compared with existing probabilistic models of power system cascades \cite{wang12,rah14,dob05,yao16}, our proposed model contributes to the effective selection of the most possible cascading failure paths with relatively low computation burdens. Given predicted failure paths, further blackouts can be prevented by executing an effective load shedding scheme. Collating multiple cascading failure paths that are most likely to occur, a robust optimization problem is formulated to compute the minimum amount of loads to shed. Geometrically, each considered cascading failure path is associated with a convex set, and the intersection of such convex domains constitute the search space to find the optimal solution of the robust optimization problem. In practice, the proposed Markov chain model and protection scheme can be adopted in the wide-area measurement and protection system to enhance the capability of preventing the cascading blackout (see Fig.\ref{fig:process}).

The main contributions of this work are summarized as follows:
\begin{enumerate}
  \item Propose a Markov chain model to predict the cascading failure paths with the consideration of uncertainties and stochastic factors.
  \item Develop a robust optimization formulation for the emergency protection to prevent cascading blackouts.
  \item Design an efficient numerical solver for the robust optimization problem using Dykstra's algorithm.
\end{enumerate}

The remainder of this paper is organized as follows: Section \ref{sec:pred} presents a Markov chain model for the prediction of cascading failure paths. Section \ref{sec:rob} provides the robust optimization formulation for the prevention of cascading blackouts, followed by an efficient numerical solver in Section \ref{sec:num}. Numerical simulations are conducted to validate the protection approach in Section \ref{sec:sim}. Finally, we draw a conclusion and discuss the future work in Section \ref{sec:con}.

\section{Prediction of Cascading Failure Paths}\label{sec:pred}
\begin{figure}
\scalebox{0.075}[0.075]{\includegraphics{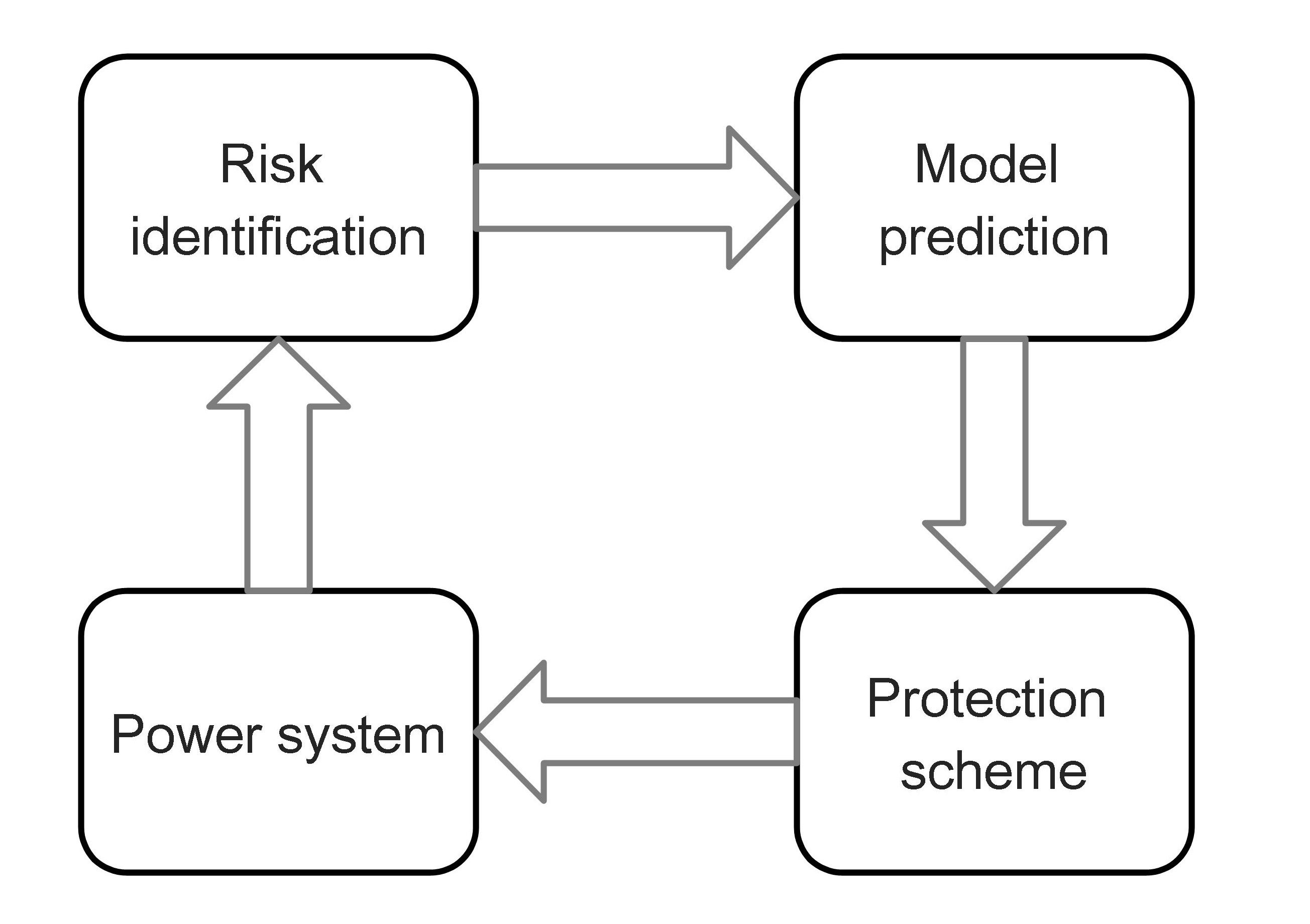}}\centering
\caption{\label{fig:process} Information flow in the wide-area measurement, protection and control system. The block of risk identification is used to detect the contingency (e.g. branch outage). And the block of model prediction produces the possible cascading failure paths. The block of protection scheme implements the protective actions (e.g. load shedding, generation control).}
\end{figure}

This section presents a Markov chain model for the prediction of cascading failure paths and an approach to handling the curse of dimensionality. Without loss of generality, we consider a power system with $n$ branches, and each branch has two connection states: ``on" and ``off", represented by the binary bits $1$ and $0$, respectively. Thus, the state of power system on branch connection can be described by a string of $n$ binary bits (see Fig.~\ref{que}).
To avoid the vagueness, it is necessary to clarify the concept of cascading step. To be precise, a cascading step refers to one topological change (e.g. one branch outage) of power networks due to branch overloads or stochastic factors such as hidden failure of relays, human errors and the weather \cite{zhai19}.

\subsection{Markov chain model}
Suppose that the probability of moving to the next power system state only depends on the present state during the cascades. Given the probability of initial states, a sequence of random variables can be generated to describe the cascading process. Then a state-dependent transition matrix $P=(p_{ij})\in R^{2^n\times2^n}$ can be created to describe the random process, and each element in $P$ is given by
$$
p_{ij}=\texttt{Prob}(s^{k+1}=j|s^k=i), \quad i,~j\in \mathcal{B}^n, \quad \mathcal{B}=\{0,1\}
$$
where $\mathcal{B}^n$ denotes the state set of Markov chain and $s^k=(s_1^k,s_2^k,...,s_n^k)\in\mathcal{B}^n$ is a vector of $n$ binary elements to characterize the branch connection (i.e. ``on" or ``off") at the $k$-th cascading step. Specifically, $p_{ij}$ is the transition probability of moving from the state $i$ to the state $j$ in one cascading step. Let $p_{hidden}$ and $p_{cont}$ denote the probabilities of branch outage due to hidden failures and the contingency, respectively. In addition, $p_{over}$ refers to the outage probability due to branch overloads. Suppose the above factors that cause branch outages are independent. And thus the outage probability of the $l$-th branch is given by
$$
\lambda_{l}=1-\left(1-p_{over,l}\right)\cdot\left(1-p_{hidden,l}\right)\cdot\left(1-p_{cont,l}\right)
$$
where $p_{over,l}$, $l\in I_n=\{1,2,...,n\}$ represents the probability of branch outage, which depends on the current or power flow on the $l$-th branch. Actually, the power flow or current on each branch is independent on the connectivity of power networks. The following theoretical results reveal the relationship between the elements of transition matrix and the probability of branch outage.

\begin{figure}
\scalebox{0.065}[0.065]{\includegraphics{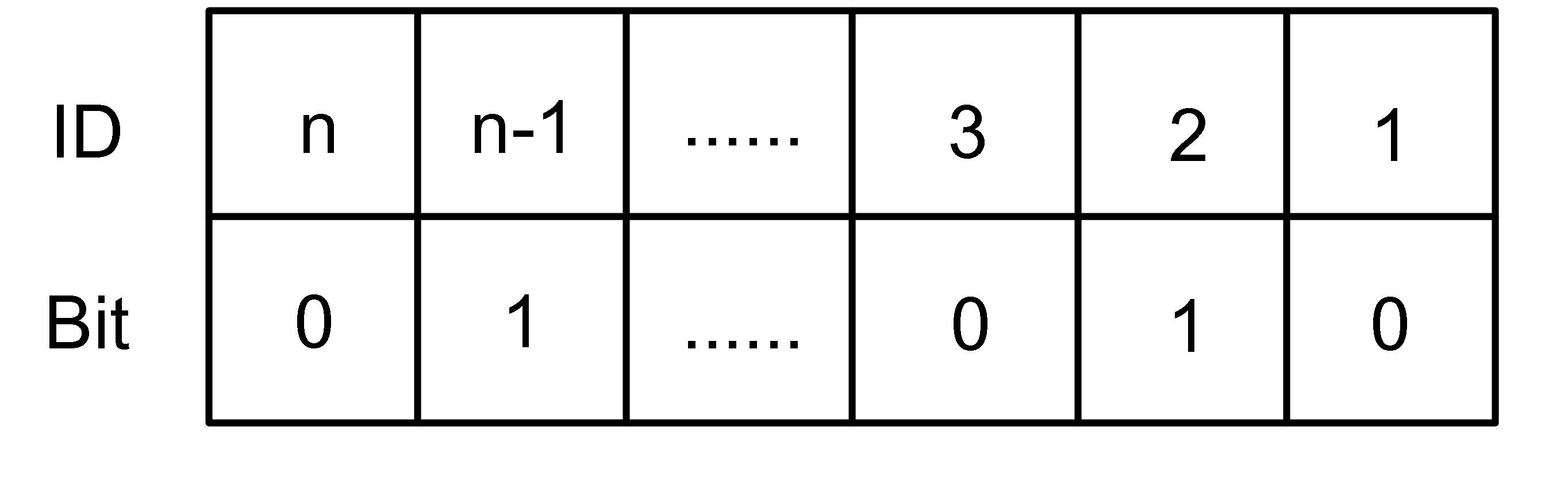}}\centering
\caption{\label{que} Branch ID and its state in a string of $n$ binary bits. The head bit on the left side stores the state of Branch $n$, and the tail bit on the right side records the state of Branch $1$.}
\end{figure}

\begin{prop}
The elements in the transition matrix $P$ satisfy
\begin{equation}\label{pij}
p_{ij}=
    \left\{
      \begin{array}{ll}
        0, & \hbox{$i\vee j\neq i$;} \\
        \prod_{l\in\Omega(i,j)}\lambda_{l}\cdot\prod_{l\in\Theta(i,j)}(1-\lambda_{l}), & \hbox{$i\vee j=i$,}
      \end{array}
    \right.
\end{equation}
where the symbol $\vee$ denotes the bitwise logical disjunction. In addition, $\Omega(i,j)$ and $\Theta(i,j)$ are two sets of branch ID as follows:
$$
\Omega(i,j)=\{l\in I_n| s^k=i, s^{k+1}=j, s_l^k=1, s_l^{k+1}=0\}
$$
and
$$
\Theta(i,j)=\{l\in I_n| s^k=i, s^{k+1}=j, s_l^k=1, s_l^{k+1}=1\}.
$$
\end{prop}

\begin{proof}
The condition $i\vee j\neq i$ implies that there exist branches, whose connection status turns to ``$1$" at the ($k+1$)-th cascading step from ``$0$" at the $k$-th cascading step. This is in contradiction with our assumption that the branch can not be reconnected any longer once it is severed. And thus we have $p_{ij}=0$. The condition $i\vee j=i$ characterizes the normal cascading process of power system. The set $\Omega(i,j)$ includes the ID numbers of severed branches in the state shift from $i$ to $j$, while the set $\Theta(i,j)$ keeps those of connected branches. Since the events of branch outage are independent, the probability of severing the branches in $\Omega(i,j)$ is $\prod_{l\in\Omega(i,j)}\lambda_{l}$ and the probability of keeping the connected branches in $\Theta(i,j)$ is
$$
\prod_{l\in\Theta(i,j)}(1-\lambda_{l}).
$$
It follows that the transition probability of moving from the state $i$ to the state $j$ is
$$
\prod_{l\in\Omega(i,j)}\lambda_{l}\cdot\prod_{l\in\Theta(i,j)}(1-\lambda_{l}),
$$
which completes the  proof.
\end{proof}
After the transition matrix $P$ is computed according to (\ref{pij}), it becomes feasible to predict the cascading failure paths of power systems (see Fig. \ref{path}). The probability distribution over states of power system on branch connection evolves as follows
$$
x^{k+1}=x^kP, \quad k\in\mathbb{N},
$$
where $\mathbb{N}$ denotes the set of nonnegative integers, and $x^k$ refers to the $2^n$-dimensional probability vector of power system states.
This allows us to obtain the probability distribution over power system states at the $k$-th cascading step
$$
x^{k}=x^0\underbrace{P\cdot P\cdot\cdot\cdot P}_{k}=x^0P^k,
$$
where $x^0$ denotes the initial probability distribution over the states of power systems, and it can be provided by the identification algorithm of branch outage in the wide-area monitoring system (WAMS) \cite{zhao14}. In practice, the transition matrix $P$ and its $k$-th power $P^k$ can be computed and recorded in advance to save time and reduce the computation burden for the online protection.

\begin{figure}
\scalebox{0.06}[0.06]{\includegraphics{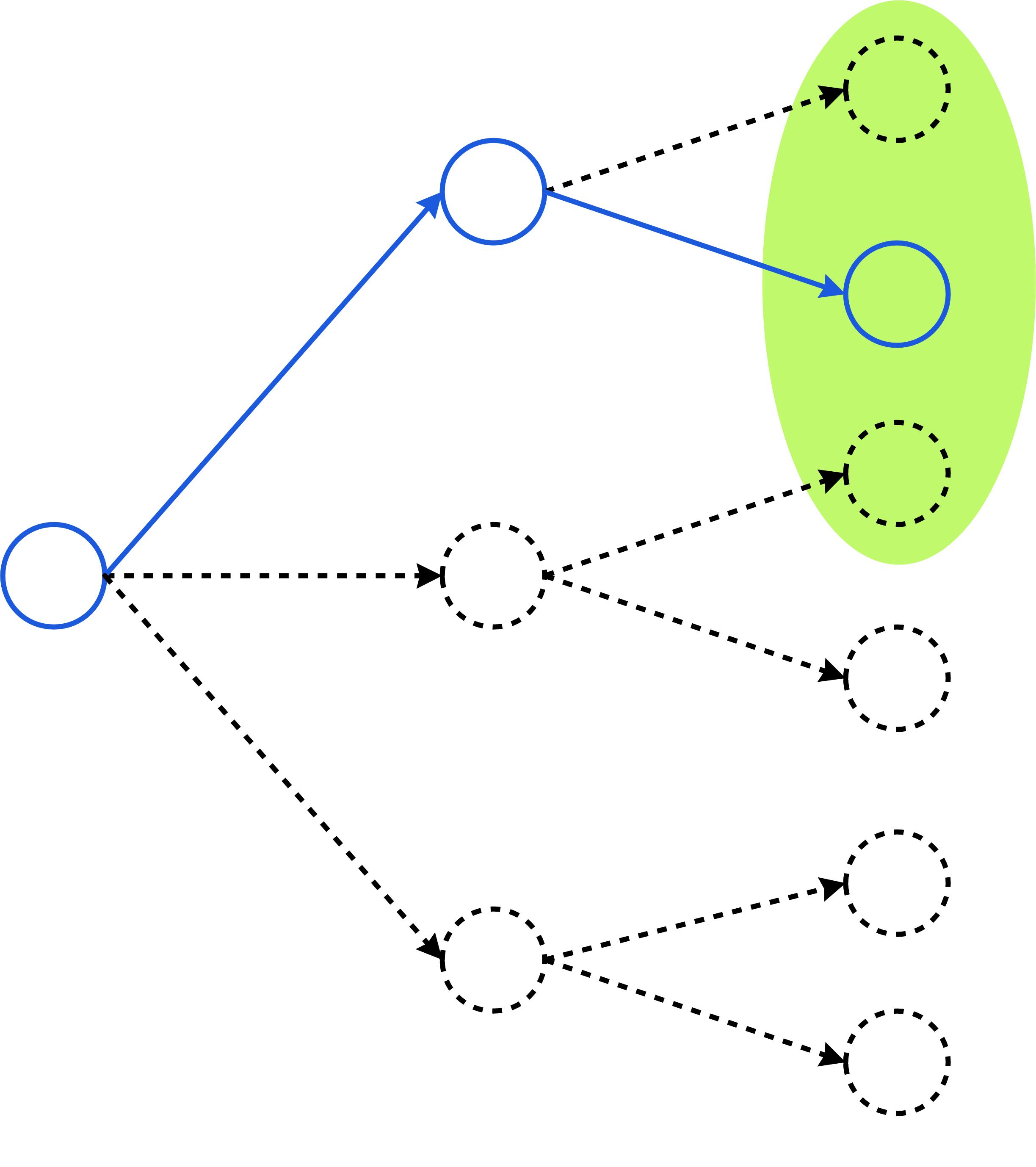}}\centering
\caption{\label{path} The predicted and actual cascading failure paths. The blue balls with solid boundaries denote the actual power system state during the cascades, while the dashed ones refer to the possible states predicted by the Markov chain model. The arrows indicate the shift of power system states. The solutions to robust optimization problem are implemented to safeguard power system states in the green ellipse against further cascades.}
\end{figure}

\subsection{Dimensionality Reduction}
As is known, the dimension of transition matrix $P$ increases exponentially with the size of power networks (i.e. number of branches). To solve the curse of dimensionality, a simplified model of power system cascades is proposed by ruling out the states with extremely low probabilities. With this simplified model, it becomes workable to predict the cascading failure paths and compute the probability of power system topologies at each cascading step.

For a power system with $n$ branches, the size of transition matrix $P$ goes up to $2^n\times2^n$. In practice, when $n$ is larger than $20$, the cost of computing $P$ becomes extremely high. Therefore, it is necessary to develop an efficient approach to estimating the probability distribution over states of power system and avoiding the curse of dimensionality. Since plenty of states occur with extremely small probabilities during the cascades, we only consider the states in the uncertainty set $\mathcal{D}_{\epsilon}$, where the elements or states satisfy a certain condition (i.e. their probabilities are larger than a certain threshold $\epsilon$). The mathematical expression of $\mathcal{D}_{\epsilon}$ is given by
\begin{equation}\label{setd}
    \mathcal{D}_{\epsilon}=\{s\in\mathcal{B}^n~|~x^k_{\nu(s)}>\epsilon,~\forall k\in I_h \},
\end{equation}
where $I_h=\{1,2,...,h\}$, and $x^k_{\nu(s)}$ denotes the probability that the cascades evolve to the state $s$ at the $k$-th cascading step. The state $s$ is described by a string of binary bits and $\nu(s)$ transforms $s$ to a decimal sequence number that ranks the state $s$ in the state space $\mathcal{B}^n$. For example, $\nu(s)$ can be defined as $\nu(s)=B2D(s)+1$, where $B2D(s)$ denotes the conversion from a binary number to an equivalent decimal number. In this way, the computation burden can be significantly relieved with the guaranteed probability of evolving into some states at the certain cascading step. For the $2^n$-dimensional row vector $x$, a vector function $\Gamma_{\epsilon}(x)$ is introduced to reset $x$ as follows
\begin{equation}\label{clfun}
    \Gamma_{\epsilon}(x)=x\cdot\texttt{diag}\left(\textbf{1}_{\epsilon}(x_1),\textbf{1}_{\epsilon}(x_2),...,\textbf{1}_{\epsilon}(x_{2^n})\right),
\end{equation}
where the operation $\texttt{diag}(x)$ obtains a square diagonal matrix with the elements of vector $x$ on the main diagonal, and $\textbf{1}_{\epsilon}(x_i)$ is the indicator function defined by
$$
\textbf{1}_{\epsilon}(x_i)=\left\{
                           \begin{array}{ll}
                             1, & \hbox{$x_i>\epsilon$;} \\
                             0, & \hbox{$x_i\leq\epsilon$.}
                           \end{array}
                         \right.
$$
where $1\leq i\leq2^n$. Thus, an iteration equation is established to estimate the probability distribution with relatively low computation burdens as follows
\begin{equation}\label{iter}
    \left\{
       \begin{array}{ll}
         \hat{x}^{0}=x^{0}, \\
         \hat{x}^{k+1}=\Gamma_{\epsilon}(\hat{x}^kP), & \hbox{$k\in\mathbb{N}$.}
       \end{array}
     \right.
\end{equation}
By using the above iteration equation, it is feasible to estimate the probability of power system states at the $k$-th cascading step.
\begin{prop}\label{prop:lower_bound}
With the iteration equation (\ref{iter}), it holds that
\begin{equation}\label{est}
\texttt{Prob}(s^k\in\mathcal{D}_{\epsilon})\geq\|\hat{x}^k\|_1
\end{equation}
where $\|\cdot\|_1$ denotes the $\textit{l}_1$ norm, and $s^k$ refers to the power system state at the $k$-th cascading step.
\end{prop}

\begin{proof}
It follows from the iteration equation (\ref{iter}) that
$$
\hat{x}^k=\Gamma_{\epsilon}(\hat{x}^{k-1}P).
$$
Considering that $\hat{x}^k_i\leq x^k_i$, $\forall i\in I_{2^n}$, we have
\begin{equation*}
\begin{split}
\texttt{Prob}(s^k\in\mathcal{D}_{\epsilon})&=\|\Gamma_{\epsilon}(x^k)\|_1 \\
&=\|\Gamma_{\epsilon}(x^{k-1}P)\|_1 \\
&\geq \|\Gamma_{\epsilon}(\hat{x}^{k-1}P)\|_1=\|\hat{x}^k\|_1,
\end{split}
\end{equation*}
which completes the proof.
\end{proof}
\section{Robust Optimization Formulation}\label{sec:rob}

In this section, we present the formulation of robust optimization for the online protection of power systems. By integrating the uncertainty of cascading failure paths with the change of branch admittance, a robust optimization problem is formulated as follows
\begin{equation}\label{opt}
\begin{split}
     &~~~~~\min_{P_b}\|P_b-P_b^0\| \\
     &s.t.~P_{e}=\texttt{diag}(\hat{Y}_p)A(A^T\texttt{diag}(\hat{Y}_p)A)^{-1^*}P_b \\
     &~~~~~\hat{Y}_p=\texttt{diag}(s)\cdot Y_p,~s\in\mathcal{D}_{\epsilon} \\
     &~~~~~P_{b,i}\in [\underline{P}_{b,i},\bar{P}_{b,i}],~i\in I_{m} \\
     &~~~~~P_{e,j}^2 \leq \sigma^2_j,~j\in I_n
\end{split}
\end{equation}
where $P_b^0$ refers to the original vector of injected power on buses before load shedding, and $P_b=(P_{b,1}, P_{b,2},...,P_{b,m})^T$ denotes the vector of injected power on buses after load shedding. $\underline{P}_{b,i}$ and $\bar{P}_{b,i}$ characterize the upper and lower bounds of the injected power on the $i$-th bus, respectively. $P_e=(P_{e,1},P_{e,2},...,P_{e,n})^T$ represents the vector of power flow on branches, and $\sigma_j$ specifies the threshold of power flow on the $j$-th branch.  $Y_p$ denotes the vector of branch susceptance. In addition, $A$ refers to the incidence matrix from branch to bus in power networks \cite{stag68}. The set $\mathcal{D}_{\epsilon}$ contains all the predicted power system states during the cascades. It is worth pointing out that the symbol $-1^{*}$ is defined as a pseudo-inverse to solve the DC power flow equation \cite{tcns19}.

The power system state $s$ determines the topology of power networks and thus affects the power flow distribution on branches.
The iteration equation (\ref{iter}) allows us to predict power system states and estimate the corresponding probabilities at a given cascading step. By implementing the solution to Problem (\ref{opt}), it is expected to prevent the cascades for the predicted power system states with a guaranteed probability.
\begin{prop}
Robust Optimization Problem (\ref{opt}) is equivalent to a convex optimization problem as follows
\begin{equation}\label{conp}
\min_{P_b\in X}\|P_b-P_b^0\|
\end{equation}
with $X=\bigcap_{s\in\mathcal{D}_{\epsilon}}X_s$ and
$$
    X_s=\left\{
  \begin{array}{c|c}
      \null & P_{e}=\texttt{diag}(\hat{Y}_p)A(A^T\texttt{diag}(\hat{Y}_p)A)^{-1^*}P_b \\
       P_b  & \hat{Y}_p=\texttt{diag}(s)\cdot Y_p,~P_{b,i}\in [\underline{P}_{b,i},\bar{P}_{b,i}] \\
      \null & P_{e,j}^2 \leq \sigma^2_j, ~i\in I_{m},~j\in I_n  \\
  \end{array}
\right\}.
$$
\end{prop}

\begin{proof}
The constraints of (\ref{opt}) can be described by the intersection of finite sets
$$
X=\bigcap_{s\in\mathcal{D}_{\epsilon}}X_s
$$
where $X_s$ is given by
\begin{equation*}
    X_s=\left\{
  \begin{array}{c|c}
      \null & P_{e}=\texttt{diag}(\hat{Y}_p)A(A^T\texttt{diag}(\hat{Y}_p)A)^{-1^*}P_b \\
       P_b  & \hat{Y}_p=\texttt{diag}(s)\cdot Y_p,~P_{b,i}\in [\underline{P}_{b,i},\bar{P}_{b,i}] \\
      \null & P_{e,j}^2 \leq \sigma^2_j,~i\in I_{m},~j\in I_n  \\
  \end{array}
\right\}.
\end{equation*}
For any $j\in {I}_n$,  the constraint function $P_{e,j}^2-\sigma^2_j$ is convex. For any $s\in \mathcal{D}_{\epsilon}$, $P_e$ is affine with respect to $P_b$, and thus $X_s$ is a convex set. Since the intersection of finite convex sets is a convex set, $X$ is a convex set as well. Thus, Robust Optimization Problem (\ref{opt}) is equivalent to the following convex optimization problem
$$
\min_{P_b\in X}\|P_b-P_b^0\|
$$
This completes the proof.
\end{proof}

\begin{figure}
\scalebox{0.07}[0.07]{\includegraphics{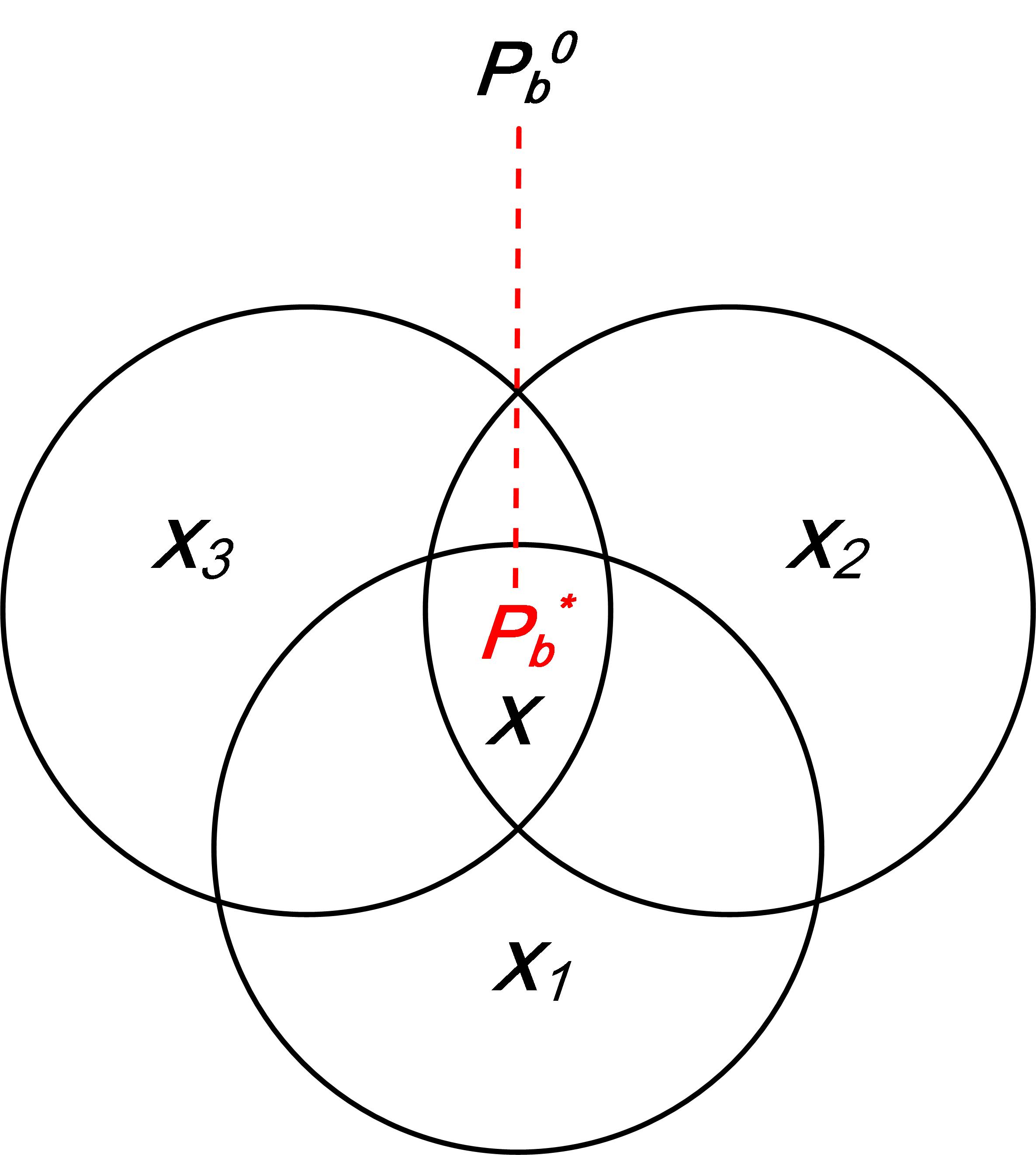}}\centering
\caption{\label{proj} Geometric interpretation of optimal solutions to the problem (\ref{opt}). For instance, there are three possible cascading failure pathes, which allows us to obtain three convex sets $X_1$, $X_2$ and $X_3$. The set $X$ denotes the intersection of $X_1$, $X_2$ and $X_3$, and the optimal solution $P^{*}_b$ to the problem (\ref{opt}) is obtained by the projection of $P_b^0$ in the set $X$.}
\end{figure}

The solutions to Robust Convex Optimization Problem (\ref{conp}) can be interpreted as the projection of a point onto the intersection of convex sets (see Fig. \ref{proj}). To determine the nonempty $X_s$, we introduce the following two convex sets on the vector of power flow
$$
\alpha=\{P_e|~P_{e,j}^2 \leq \sigma^2_j,~j\in I_n\}
$$
and
$$
\beta_s=\left\{
  \begin{array}{c|c}
      \null & \hat{Y}_p=\texttt{diag}(s)\cdot Y_p \\
       P_e  & P_{b,i}\in [\underline{P}_{b,i},\bar{P}_{b,i}],~i\in I_{m} \\
      \null & P_{e}=\texttt{diag}(\hat{Y}_p)A(A^T\texttt{diag}(\hat{Y}_p)A)^{-1^*}P_b \\
  \end{array}
\right\}.
$$
It follows that $X_s\neq\emptyset$ if and only if $\texttt{dist}(\alpha,\beta_s)=0$. Actually, the distance between two convex sets $\texttt{dist}(\alpha,\beta_s)$ can be obtained by solving the following convex optimization problem: $\min\|x-y\|$ subject to
$x\in\alpha$ and $y\in\beta_s$. The efficient numerical algorithms are available to compute the distance between two convex sets and check the non-null of their intersection \cite{lla00}. To verify that $X$ is a non-null set, we need to find a point in $X$. Actually, it is associated with the set intersection problem (SIP), and the method of alternating projections can be adopted to solve the SIP \cite{bau11}. Thus, we have the following theorem

\begin{theorem}\label{theo}
For the non-null set $X$, the solutions to Problem (\ref{opt}) can prevent the further cascades with the probability as least $\|\hat{x}^k\|_1$ at the $k$-th cascading step.
\end{theorem}

\begin{proof}
Proposition \ref{prop:lower_bound} provides the lower bound of the probability that power system states are in the set $\mathcal{D}_{\epsilon}$ at the $k$-th cascading step. Since $X$ is a non-null set that is the intersection of convex
sets $X_s$, $s\in\mathcal{D}_{\epsilon}$, the optimal solutions to Problem (\ref{opt}) can allow for all the states in $\mathcal{D}_{\epsilon}$. This completes the proof.
\end{proof}

\section{Numerical Solver using Dykstra's Algorithm}\label{sec:num}
Essentially, Robust Optimization Problem (\ref{opt}) is the best approximation problem (BAP), which considers the projection of a point onto the intersection of a number of convex sets \cite{jef16}. Note that Dykstra's algorithm ensures the convergence of optimal solutions to the BAP \cite{boy86}. Thus, a numerical solver based on Dykstra's algorithm is developed in Table \ref{dyk} in order to find the optimal solution to Problem (\ref{opt}). Before implementing the numerical solver, it is required to specify the iteration number $N$ and initial values in each convex set $X^{(i)}$, $i\in I_{|\mathcal{D}_{\epsilon}|}$. Actually, the index $i$ of the set $X^{(i)}$ denotes the sequence number of $X_s$ in the set $\mathcal{D}_{\epsilon}$. Here $|\mathcal{D}_{\epsilon}|$ denotes the cardinality of the set $\mathcal{D}_{\epsilon}$. And the goal is to find a point in the intersection set $X$ that is closest to the point $P^0_b$. Then the initial deviation $E^{(0)}_{|\mathcal{D}_{\epsilon}|}$ is computed for the correction of projection in each set $X^{(i)}$. An iteration loop is established to implement a sequence of projections onto
each convex set and the corrections, and the loop terminates once the iteration number is reached.

\begin{table}
 \caption{\label{dyk} Numerical Solver using Dykstra's Algorithm.}
 \begin{center}
 \begin{tabular}{lcl} \hline
 \textbf{Initialize:} $N$, $P^{(0)}_{b,i}\in[\underline{P}_{b,1},\bar{P}_{b,1}]\times\cdot\cdot\cdot[\underline{P}_{b,m},\bar{P}_{b,m}]$, $\forall i\in I_{|\mathcal{D}_{\epsilon}|}$ \\
 \textbf{Goal:}  $\argmin_{P_b\in X}\|P_b-P_b^0\|$ \\ \hline
  1: ~$E^{(0)}_{|\mathcal{D}_{\epsilon}|}=P_b^0-\sum_{i=1}^{|\mathcal{D}_{\epsilon}|}P^{(0)}_{b,i}$ \\
  2: ~\textbf{for}~$k=1$~\textbf{to}~$N$ \\
  3: ~~~~~~~$E^{(k)}_0=E^{(k-1)}_{|\mathcal{D}_{\epsilon}|}$ \\
  4: ~~~~~~~\textbf{for}~$i=1$ \textbf{to} $|\mathcal{D}_{\epsilon}|$ \\
  5: ~~~~~~~~~~~~~$Z_i^{(k)}=E_{i-1}^{(k)}+P_{b,i}^{(k-1)}$  \\
  6: ~~~~~~~~~~~~~$E_i^{(k)}=P_{X^{(i)}}(Z_i^{(k)})$  \\
  7: ~~~~~~~~~~~~~$P_{b,i}^{(k)}=Z_i^{(k)}-E_i^{(k)}$ \\
  8: ~~~~~~~\textbf{end for} \\
  9: ~\textbf{end for} \\ \hline
 \end{tabular}
 \end{center}
\end{table}

\section{Simulation and Validation}\label{sec:sim}
In this section, we validate the proposed robust optimization approach for power system protection on IEEE 118 bus system. Suppose the initial malicious contingencies can be identified by power systems, and it immediately triggers the power system cascades. In the simulation, the Markov chain model enables us to determine the most probable cascading failure paths. By implementing the numerical solver in Table \ref{dyk}, the optimal solutions for load shedding and generation control are obtained to terminate the cascades at
a specified cascading step. Note that our focus is on the power system protection at the steady-state progression of cascades, where
the progression of the cascade events is slow, and the balance between the power generation and consummation is not lost \cite{lu96}.

\subsection{Parameter setting}
For simplicity, the probability of branch outage due to contingencies is identical for each branch with $p_{cont}=10^{-4}$. The probability of branch outage due to hidden failure is dependent on the location of outage branches. It is demonstrated that one branch is more likely to fail if it is connected to the outage branches through a common bus. Thus, we specify $p_{hidden}=10^{-2}$ for the branches connected to one outage branch and $p_{hidden}=10^{-4}$ for those without connecting to any outage branch. As for the probability of branch outage due to overloads $p_{over}$, we introduce a ratio $r_i$ between the actual power flow on Branch $i$ and its threshold, and $p_{over}$ for Branch $i$ is computed as follows:
$$
p_{over}=\left\{
           \begin{array}{ll}
             0, & \hbox{$r_i<0.8$;} \\
             1, & \hbox{$r_i>1.05$;} \\
             4r_i-3.2, & \hbox{otherwise.}
           \end{array}
         \right.
$$
In addition, the DC power flow equation is employed to compute the power flow on branches. Per unit values are adopted with the base value of power $100$ MVA. It is assumed that the threshold of power flow on each branch is two times larger than the normal power flow on the corresponding branch. The value of $\epsilon$ is equal to $0.1$ in order to select the most probable cascading failure paths. The iteration number $N$ in Table \ref{dyk} is $50$ to implement the numerical solver.

\subsection{Validation and discussion}
\begin{figure}\centering
 {\includegraphics[width=0.49\textwidth]{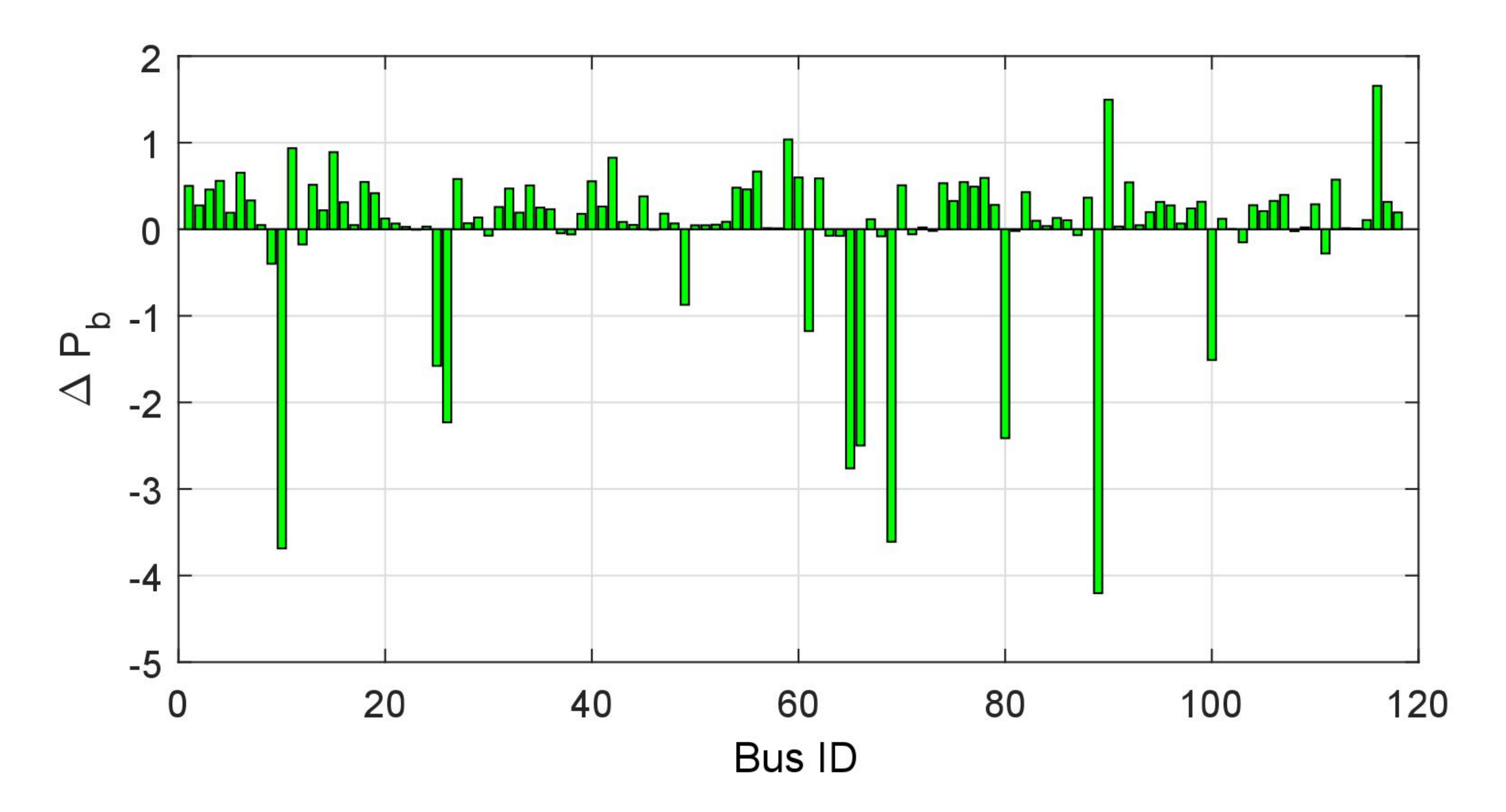}}
 \caption{\label{n50} Distribution of changes of injected power on each bus after implementing the numerical solver based on Dykstra's Algorithm in Table \ref{dyk}.}
\end{figure}

With the proposed Markov chain model and the above parameter setting, we can identify the most probable cascading paths at the first several cascading steps by severing Branch $8$ with the probability $0.6$ or severing Branch $4$ with the probability $0.4$ as the initial contingency, respectively. Thus, we can obtain two most probable cascading failure paths as follows. Path $A$: $8\rightarrow(21,36,37)\rightarrow32\rightarrow(38, 50,54)$ and Path $B$: $4\rightarrow12\rightarrow5\rightarrow37$, where each number denotes the ID of outage branch at the corresponding cascading step, and the arrow indicates the shift of cascading steps. Our goal is to terminate the further cascades at the $3$rd cascading step by implementing the solutions of numerical solver in Table \ref{dyk}. Let $\Delta P_b=P^{*}_b-P^0_b$ denote the changes of injected power on buses after taking protective actions according to the solutions of numerical solver. Figure \ref{n50} shows the distribution of $\Delta P_b$ for each bus after implementing the solver in Table \ref{dyk}. It is demonstrated that these two most probable cascading failure paths can be terminated at the $3$rd cascading step by adjusting the injected power on buses through load shedding and generation control. According to Theorem \ref{theo}, the cascades can be prevented with the probability at least $0.92$.

Actually, there is a tradeoff between the optimal adjustment of injected power on buses and the effective termination of multiple cascading failure paths. In other words, the larger changes of injected power on buses have to be made for terminating more cascading failure paths.  In the worst cases, it might be infeasible to prevent the further cascades by only adjusting the injected power on buses at a given cascading step (e.g. $X=\emptyset$). Thus, other remedial actions should be taken to protect power systems against blackouts (e.g. proactive line tripping, adjustment of branch impedance, etc).

\section{Conclusions}\label{sec:con}
This paper investigated the problem of preventing cascading blackout of power systems with uncertainties. A Markov-chain model was proposed to deal with the uncertainties and predict the cascading failure paths. In particular, a robust convex optimization problem was formulated to terminate the cascades by optimally shedding loads at the possible power system states. Moreover, an efficient numerical solver based on Dykstra's algorithm was adopted to solve the proposed robust optimization problem. In theory, we provide the lower bound for the probability of preventing the cascading blackouts of power systems. Numerical simulations were conducted to validate the proposed robust optimization approach and the theoretical lower bound on the probability of terminating the cascades. Future work may include the calibration of the proposed Markov chain model by using the statistical data on real power system cascades as well as the validation of the proposed protection scheme \cite{dob16}. In addition, more efforts will be taken on the development of efficient numerical algorithms to solve the robust optimization problem so that the cascades can be terminated in short time.


\begin{thebibliography}{1}

\bibitem{mcl09} J. McLinn, ``Major Power Outages in the US, and around the World," {\sl Annual Technology Report of IEEE Reliability Society}, 2009.

\bibitem{lu06} W. Lu, Y. Besanger, E. Zamai, and D. Radu,  ``Blackouts: Description, analysis and classification," {\sl Proceedings of the 6th WSEAS International Conference on Power Systems}, Lisbon, Portugal, pp. 429-434, 2006, September.

\bibitem{arn97} S. Arnborg, G. Andersson, D. Hill, and I. Hiskens, ``On undervoltage load shedding in power systems," {\sl International Journal of Electrical Power \& Energy Systems}, 19(2), pp. 141-149, 1997.

\bibitem{sur67} ``Survey of underfrequncy relay tripping of load under emergency conditions -- IEEE committee report," presented at the Summer Power Meeting, Portland, OR, 1967, Paper 31 TP 67-402.

\bibitem{sta74} ``A status report on methods used for system preservation during underfrequency conditions," in Summer Meeting Energy Resources Conf., Anaheim, CA, 1974, Paper 74 310-9

\bibitem{hew04} L. Hewitson, B. Mark, and B. Ramesh, {\sl Practical Power System Protection}, Elsevier, 2004.

\bibitem{and96} P. Anderson, and B. LeReverend, ``Industry experience with special protection schemes," {\sl IEEE Transactions on Power Systems}, 11(3), 1166-1179, 1996.

\bibitem{sps98} (1998, Sep.) Maintaining reliability in a competitive U.S. electricity industry: Final report of the North American Electric Reliability Council (NERC) Task Force on Electric System Reliability. [Online]. Available: http://www.nerc.com/~filez/reports.html.

\bibitem{beg05} M. Begovic, D. Novosel, D. Karlsson, C. Henville, and G. Michel, ``Wide-area protection and emergency control," {\sl Proceedings of the IEEE}, 93(5), pp. 876-891, 2005.

\bibitem{nuq05} R. Nuqui, and A. Phadke, ``Phasor measurement unit placement techniques for complete and incomplete observability," {\sl IEEE Transactions on Power Delivery}, 20(4), pp. 2381-2388, 2005.

\bibitem{zhai19} C. Zhai, H. Zhang, G. Xiao, and T. Pan, ``A model predictive approach to protect power systems against cascading blackouts," {\sl International Journal of Electrical Power \& Energy Systems}, 113, pp. 310-321, 2019	

\bibitem{wang12} Z. Wang, A. Scaglione and R. J. Thomas, ``A Markov-transition model for cascading failures in power grids," The 45th Hawaii International Conference on System Science, 2012, pp. 2115-2124.

\bibitem{rah14} M. Rahnamay-Naeini, Z. Wang, N. Ghani, A. Mammoli, and M. M. Hayat, ``Stochastic analysis of cascading-failure dynamics in power grids," {\sl IEEE Transations on Power Systems}, 29(4), pp. 1767-1779, 2014.

\bibitem{dob05} I. Dobson, B. Carreras, and D. Newman, ``A loading-dependent model of probabilistic cascading failure," {\sl Probabil. Eng. Inf. Sci.}, 19(1),  pp. 15-32, 2005.

\bibitem{yao16} R. Yao, S. Huang, K. Sun, F. Liu, X. Zhang, S. Mei, W. Wei, and L. Ding, ``Risk assessment of multi-timescale cascading outages based on Markovian tree search," {\sl IEEE Transactions on Power Systems}, 32(4), pp. 2887-2900, 2016.

\bibitem{zhao14} Y. Zhao, J. Chen,  A. Goldsmith, and H. Poor, ``Identification of outages in power systems with uncertain states and optimal sensor locations," {\sl IEEE Journal of Selected Topics in Signal Processing}, 8(6), pp. 1140-1153, 2014.

\bibitem{stag68} G. Stagg and H. Ahmed, {\sl Computer Methods in Power System Analysis}, McGraw-Hill, 1968.

\bibitem{tcns19} C. Zhai, H. Zhang, G. Xiao, and T. Pan, ``An optimal control approach to identify the worst-case cascading failures in power systems," {\sl IEEE Transactions on Control of Network Systems}, DOI: 10.1109/TCNS.2019.2930871, July 2019.	

\bibitem{lla00} B. Llanas, M. De Sevilla, and V. Feliu, ``An iterative algorithm for finding a nearest pair of points in two convex subsets of Rn," {\sl Computers \& Mathematics with Applications}, 40(8-9), pp. 971-983, 2000.

\bibitem{bau11} H. Bauschke and P. Combettes, {\sl Convex Analysis and Monotone Operator Theory in Hilbert Spaces}, Springer, New York, 2011.

\bibitem{jef16} C. Jeffrey Pang, ``The supporting halfspace-quadratic programming strategy for the dual of the best approximation problem," {\sl SIAM Journal on Optimization}, 26(4), pp. 2591-2619, 2016.

\bibitem{boy86} J. Boyle, R. Dykstra, ``A method for finding projections onto the intersection of convex sets in Hilbert spaces," {\sl Lecture Notes in Statistics}, 37. pp. 28-47, 1986.

\bibitem{lu96} W. Lu, Y. Besanger, E. Zama, and D. Radu, ``Blackouts: Description, Analysis and Classification," {\sl Network}, 2, p. 14, 1996.

\bibitem{dob16} I. Dobson, B. A. Carreras, D. E. Newman, and J. M. Reynolds-Barredo, ``Obtaining Statistics of Cascading Line Outages Spreading in an Electric Transmission  Network  From  Standard  Utility  Data," {\sl IEEE Transactionson Power Systems}, 31(6), pp. 4831-4841, 2016.

\end{thebibliography}
\end{document}